\documentclass[10pt,journal]{IEEEtran}

\usepackage{epstopdf}
\usepackage{setspace}
\usepackage{amsmath}
\usepackage{amssymb}
\usepackage{stfloats}
\usepackage{amsthm}
\usepackage{multirow}
\usepackage{amsfonts}
\usepackage{color}
\usepackage{graphicx}
\usepackage{placeins}
\usepackage{lipsum}
\usepackage{subfigure}  
\usepackage[]{algorithmicx}
\usepackage{algpseudocode,algorithm}
\usepackage{cite}
\usepackage{mathtools}
\usepackage{stmaryrd}
\usepackage[mathscr]{euscript}
\usepackage{array}
\usepackage{mathrsfs}

\usepackage{soul}

\newtheorem{theorem}{Theorem}
\newtheorem{lemma}{Lemma}

\newcommand{\RN}[1]{%
  \textup{\uppercase\expandafter{\romannumeral#1}}%
}
\usepackage{authblk}

\author{   Hila Naaman, {\it Student Member, IEEE}, Nimrod Glazer {\it Member, IEEE},  Moshe Namer, Daniel Bilik, Shlomi Savariego, and Yonina C. Eldar, {\it Fellow, IEEE}
\thanks
	\thanks{\scriptsize All the authors are with the Faculty of Math and Computer Science, Weizmann Institute of Science, Israel. Email: hila.naaman@weizmann.ac.il}	
	\thanks{\scriptsize Parts of this work were presented at the international Symposium on Information Theory, ISIT, July 2022.}
	\thanks{\scriptsize 
  This research was partially supported by the European Union’s Horizon 2020 research and innovation program under grant No. 101000967-ERC-CoDeS, by the Israel Science Foundation under grant no. 0100101, and by the QuantERA grant  C’MON-QSENS.}
}

\makeatletter
\let\NAT@parse\undefined
\makeatother
\usepackage[colorlinks,citecolor=green,urlcolor=blue,bookmarks=false,hypertexnames=true]{hyperref}  

\begin{document}

\title{Hardware Prototype of a
Time-Encoding Sub-Nyquist ADC}
\maketitle

\begin{abstract}
Analog-to-digital converters (ADCs) are key components of digital signal processing. Classical samplers in this framework are controlled by a global clock.
At high sampling rates, clocks are expensive and power-hungry, thus increasing the cost and energy consumption of ADCs. It is, therefore, desirable to sample using a clock-less ADC at the lowest possible rate.
An integrate-and-fire time-encoding machine (IF-TEM) is a time-based power-efficient asynchronous design that is not synced to a global clock. 
Finite-rate-of-innovation (FRI) signals, ubiquitous in various applications, have fewer degrees of freedom than the signal's Nyquist rate, enabling sub-Nyquist sampling signal models.
This work proposes a power-efficient IF-TEM ADC architecture and demonstrates sub-Nyquist sampling and FRI signal recovery.
Using an IF-TEM, we implement in hardware the first sub-Nyquist time-based sampler.
We offer a feasible approach for accurately estimating the FRI parameters from IF-TEM data.
The suggested hardware and reconstruction approach retrieves FRI parameters with an error of up to -25dB while operating at rates approximately 10 times lower than the Nyquist rate, paving the way to low-power ADC architectures.
\end{abstract}

%
\IEEEpeerreviewmaketitle

\section{Introduction}
Analog-to-digital converters (ADCs) are electronic hardware components that facilitate the digital processing of signals and communication between computers and the physical world \cite{eldar2015sampling,unser2000sampling}. Traditional ADCs, also known as synchronous ADCs, are controlled by a global clock that operates at a rate that meets the Nyquist rate, requiring the acquisition of samples at intervals of $1/2W$ seconds for signals with a frequency no greater than $W$Hz \cite{nyquist1928certain}. However, synchronous ADCs have several limitations that may make them less suitable for certain applications. One limitation is high power consumption due to the continuous clock signal, which can be a significant disadvantage in energy-constrained systems such as battery-powered devices \cite{piyare2018demand}. Another limitation is the need for a stable and accurate clock signal, which becomes more challenging to achieve as the sampling rate in a high speed system increases, especially in noisy or interference-prone environments \cite{shake2004simple,siddharth2018low,kinniment1999low}. Synchronous ADCs also require complex clock circuits, which increase the complexity of the design and implementation \cite{rastogi2011integrate,akopyan2006level}.
Consequently, there is a need for innovative ADCs that address these limitations by reducing both power consumption and sampling rate.

The integrate-and-fire time encoding machine
(IF-TEM), an asynchronous energy-efficient event-driven sampler, is a promising alternative to conventional ADCs \cite{alvarado2011integrate,koscielnik2008asynchronous,miskowicz2010efficiency,rastogi2011integrate,sayiner1996level}.
In this architecture, no global clock is required, making the IF-TEM sampler low energy. Furthermore, compared to its traditional amplitude-based ADCs, TEMs use extremely simple, entirely analog, low-power, and small size encoders~\cite{alvarado2011integrate,tsividis2010event,miskowicz2015reducing,carvalho2020hardware}.
An IF-TEM integrates an input signal and then compares the integral to a threshold; if the threshold is met, the time instances are recorded ~\cite{lazar2004perfect,lazar2004time,adam2019multi,adam2020sampling,alexandru2019reconstructing,naaman2021fri}.
The IF-TEM sampler has been utilized for ultra-wide-band (UWB) communications \cite{maravic2004channel}, remote sensing \cite{davies2021advancing,simeone2019learning}, heart activity monitoring \cite{nallathambi2013integrate,alvarado2012time}, event-based cameras (also referred to as neuromorphic cameras) \cite{gallego2019event,alexandru2021time,lichtsteiner2008128,rebecq2019events} and other applications such as spiking neural network (SNN) interpretations, leading to better knowledge of how to utilize neuromorphic hardware and replace power-hungry ADCs \cite{adam2022time}.

In \cite{lazar2004perfect}, it was shown that bandlimited signals sampled by an IF-TEM can be perfectly recovered if the average sampling rate of the IF-TEM is higher than the signal's Nyquist sampling rate.
By requiring the bandwidth to be inversely proportional to the interval between time instances, the reconstruction of the original signal closely resembles the reconstruction of a bandlimited signal sampled with irregular amplitude samples.
In \cite{kong2011analog}, it was shown that a spectrally sparse signal could be recovered when the average IF-TEM sampling rate is below the Nyquist rate with high probability. The introduced TEM was affected by frequency-dependent quantization noise, which was most significant at high-frequency input signals.
Reconstruction of signals from time encoding has been generalized for signals in shift-invariant spaces \cite{gontier2014sampling}, and finite rate of innovation (FRI) signals \cite{alexandru2019reconstructing,rudresh2020time,naaman2021fri,hilton2021time}. 

FRI signals are characterized by a small number of degrees of freedom that permit sub-Nyquist sampling \cite{vetterli2002sampling,eldar2015sampling}.
Due to their prevalence in numerous scientific applications, such as radar \cite{bar2014sub,bajwa2011identification}, ultrasound \cite{tur2011innovation,wagner2012compressed, mulleti2014ultrasound}, time-domain optical-coherence tomography (TDOCT) \cite{blu2002new}, and light detection
and ranging (LIDAR) \cite{castorena2015sampling}, sampling and recovery of FRI signals, particularly through the use of IF-TEMs, is of great interest \cite{naaman2021fri,alexandru2019time,rudresh2020time,naaman2022uniqueness}. 
Most of the FRI sampling literature focuses on reducing the ADC’s sampling rate by using the signal structure. It ignores other
aspects of the ADC, such as its power consuming clock \cite{tur2011innovation,vetterli2002sampling,dragotti2007sampling,mulleti2017paley}. We address the issue of the synchronous ADCs' power consumption by utilizing the asynchronous IF-TEM sampler, which is energy-efficient.

Time-based sampling of FRI signals can be performed similarly to conventional FRI sampling techniques, such as kernel-based sampling \cite{rudresh2020time,alexandru2019reconstructing, hilton2021time,naaman2021fri,naaman2021time,naaman2022uniqueness}. 
The authors in \cite{naaman2021fri} provided theoretical guarantees for the sampling and recovery of FRI signals using an IF-TEM, and proposed a sampling method that is more robust in the presence of noise than existing techniques.
Our work introduces a low-power IF-TEM ADC hardware that demonstrates sub-Nyquist sampling and FRI signal recovery based on the approach in \cite{naaman2021fri}.
We use hardware-measured data with time instances perturbations up to 35$ms$.
The jittered time instances are modeled as $t_n^{\prime} =t_n+ \epsilon_n$, where $t_n$ are the ideal time instances and we model the jitter noise as i.i.d. uniformly distributed
$\epsilon_n\overset{\mathrm{iid}}{\sim} \mathcal{U}[-\frac{\sigma}{2},\frac{\sigma}{2}]$.
Based upon these assumptions and the measured time instances from the hardware, it appears that the noise level $\sigma$ fluctuates between 15-70 $ms$. 
As present reconstruction techniques are incapable of dealing with such large perturbations, we modify the method of \cite{naaman2021fri} to introduce robustness in the presence of large timing noise.

	\begin{figure}[t!]
		\centering
		\includegraphics[width=0.48\textwidth]{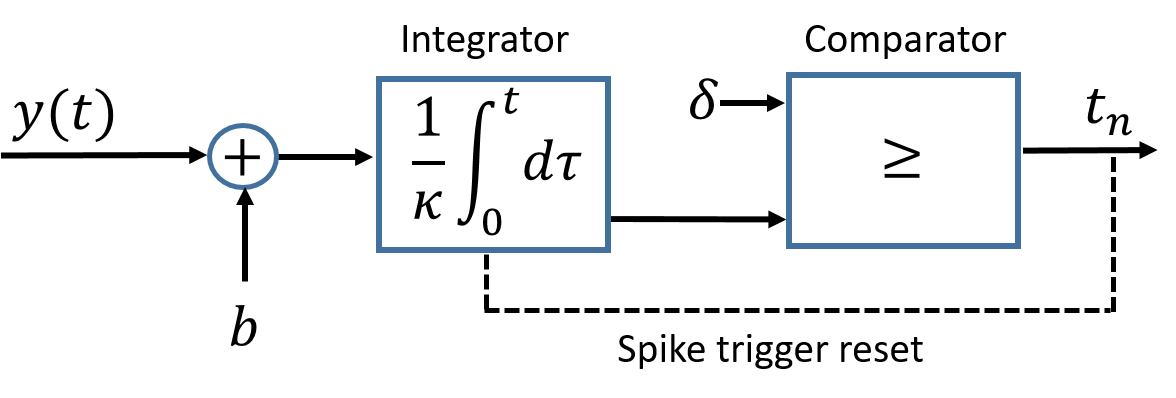}
		\caption{Time encoding machine with spike trigger reset. The input is biased by $b$, scaled by $\kappa$, and integrated. A time instant is recorded when the threshold $\delta$ is reached, after which the value of the integrator resets.}
		\label{fig:Vetterli_TEM}
	\end{figure}


Our contribution is twofold: first, we introduce a robust sub-Nyquist sampling and reconstruction technique; then, we present a hardware implementation of sub-Nyquist TEM sampling of FRI signals.
Prior to acquiring timing information with the IF-TEM, similarly to \cite{naaman2021fri}, the signal is prefiltered using a sampling kernel which eliminates the zeroth frequency component of the signal for robust recovery.
Our reconstruction method relies on the sampling kernel selection as well as
introducing a new forward model to improve recovery from noisy hardware data. Compared to our previous results [16], here we present a simpler, straightforward proof for the recovery guarantees, which is based on using a partial sum of the measurements, resulting in more stable reconstruction.
We demonstrate that in the presence of noise, the proposed reconstruction technique outperforms the method in \cite{naaman2021fri}.
Then, we present the FRI-TEM hardware prototype that can be employed in low-power time-of-flight applications.
The hardware components are designed to accommodate a broad spectrum of FRI signal frequencies.
The two primary components of the hardware are an integrator and a reset function. As long as the input signal is positive, the integrator capacitor must operate in its linear domain, which is continually charged or increasing.
In addition, the IF-TEM thresholding requires a means for a rapid reset. These are achieved by incorporating a differentiator and a FET into the reset function.

We demonstrate the capabilities of the system via several FRI signals. Prior to the IF-TEM system, a band-pass filter is employed as the sampling kernel. The filter eliminates unnecessary signal information and enables sub-Nyquist sampling. The designed hardware samples the filtered signal, resulting in time instances. One method of recording time instances or their differences is to use an oscilloscope. For estimating the FRI parameters, the Fourier coefficients are computed using our suggested algorithm, and the parameters are subsequently estimated using the annihilating filter technique. We demonstrate that it is possible to estimate FRI parameters with sub-Nyquist samples, taken at approximately 10 times the rate of innovation, which is significantly lower than the Nyquist rate of the signal.

\begin{figure}[t!]
		\centering
		\includegraphics[width=0.5\textwidth]{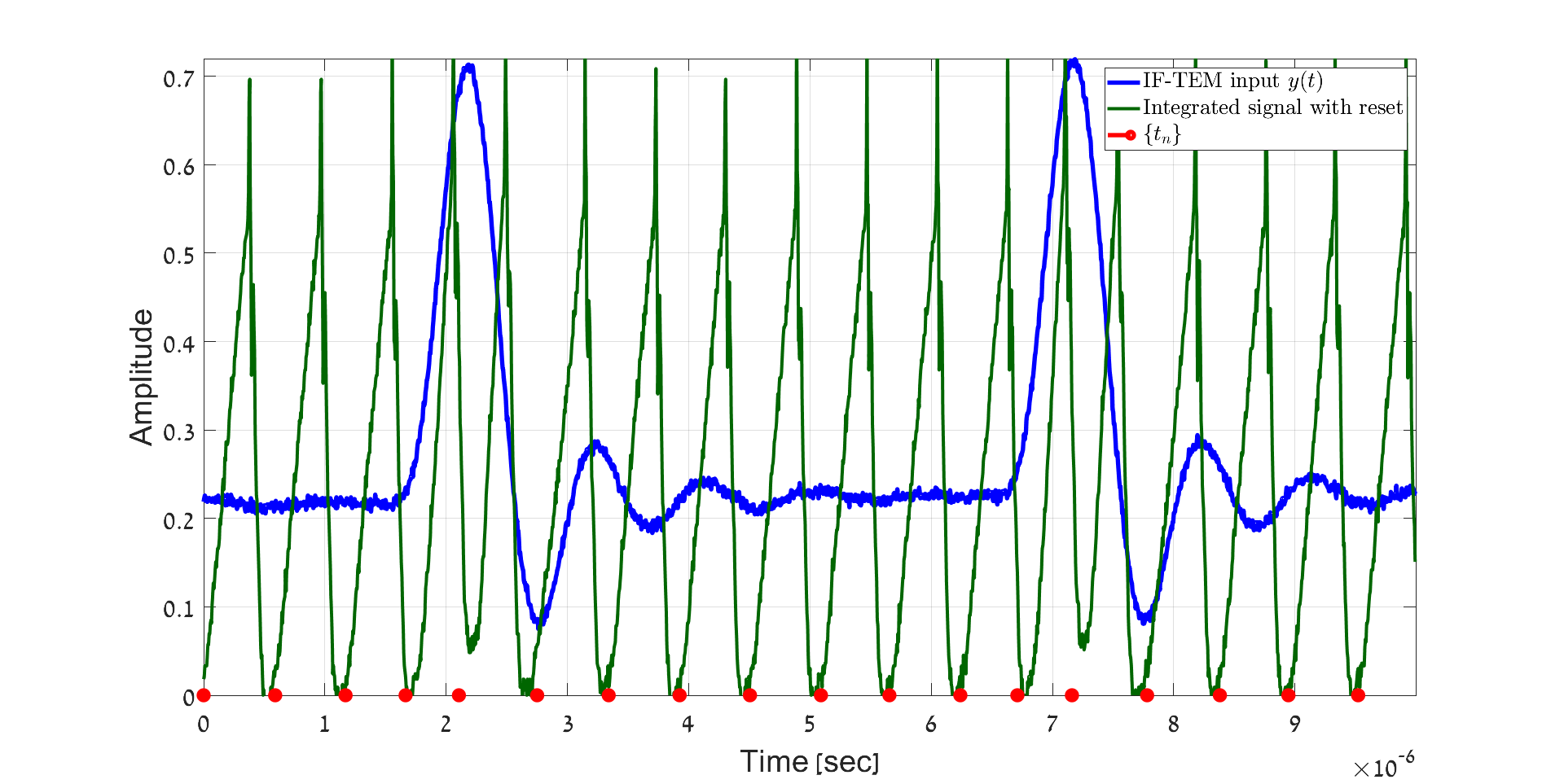}
		\caption{Our IF-TEM hardware sampling: the IF-TEM input signal $y(t)$ (blue), the integrator output (green), and the IF-TEM output time instances (red).}
		\label{fig:HW_integrator}
	\end{figure}
	
The rest of the paper is organized as follows. In Section \ref{sec:encode}, we formulate the problem of sampling and recovering an FRI signal using an IF-TEM, and discuss some background results. In Section \ref{sec:REC}, we present our robust reconstruction algorithm together with simulation results.
In Section \ref{sec:HW_spec},
we justify the required hardware specs and comment on the circuit challenges, followed by a detailed analog board's design work specifications.
Experimental hardware results of IF-TEM sub-Nyquist sampling and reconstruction are shown in Section \ref{sec:HW_experiments}.
Finally, we conclude the paper in Section \ref{sec:Conclusions}.

\section{Preliminary results Problem Formulation}
\label{sec:encode}
In this section, we review some previously established results in time encoding and FRI, followed by our formulation of the theoretical problem of FRI sampling and reconstruction utilizing an IF-TEM sampler.
\subsection{Time Encoding Machine}
\label{subsec:TEM}
We consider an IF-TEM whose operating principle is the same as in \cite{naaman2021fri} (see Fig. \ref{fig:Vetterli_TEM}). The input to the IF-TEM is a bounded signal $y(t)$, and the output is a series of spikes or time instances. An IF-TEM is parameterized by positive real numbers $b$, $\kappa$, and $\delta$.
A bias $b$ is added to a $c$-bounded signal $y(t)$ such that $|y(t)|\leq c<b<\infty$, and the sum is integrated and scaled by $\frac{1}{\kappa}$. When the resulting signal reaches the threshold $\delta$, the time instant $t_n$ is recorded, and the integrator is reset. 
The IF-TEM process is repeated to record subsequent time instants, i.e., if a time instant $t_n$ was recorded, the next time instant $t_{n+1}$ satisfies
\begin{align}
   \frac{1}{\kappa}\int_{t_n}^{t_{n+1}} (y(s)+b)\, ds = \delta.
   \label{eq:y_t_relation_st}
\end{align}

Fig. \ref{fig:HW_integrator} depicts the operational output of our IF-TEM hardware implementation using real data. The integrator constant $\kappa$ is determined from the integrator circuit hardware as demonstrated in Fig. \ref{fig:integrator}.  
The time encodings $\{ t_n, n\in\mathbb{Z} \}$ form a discrete representation of the analog signal $y(t)$ and the objective is to reconstruct $y(t)$ from them. Typically, reconstruction is performed using an alternative set of discrete representations $\{y_n, n \in \mathbb{Z}\}$ defined as
\begin{equation}
y_n \triangleq \int_{t_n}^{t_{n+1}}y(s)\, ds =  -b(t_{n+1}-t_n)+\kappa\delta.
\label{eq:trigger0}
\end{equation}
The measurements $\{y_n, n \in \mathbb{Z}\}$ are derived from the time encodings $\{ t_n, n\in\mathbb{Z} \}$ and IF-TEM parameters $\{b, \kappa, \delta\}$. 
\begin{figure}
    \centering
    \includegraphics[width=0.48\textwidth]{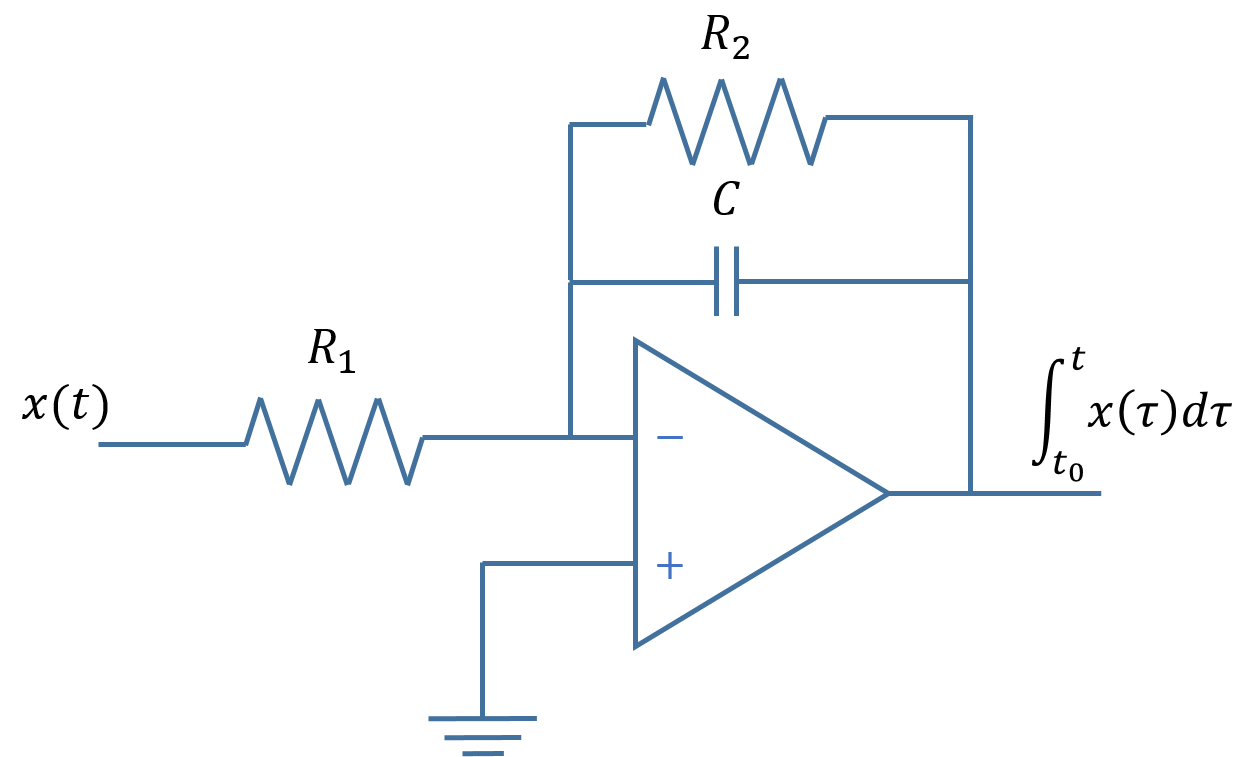}
    \caption{Hardware integrator circuit. Our hardware implementation is comprised of an operational amplifier, a capacitor $\emph{C}$ and resistors $\emph{R}_1$ and $\emph{R}_2$.}
    \label{fig:integrator}
\end{figure}
Using \eqref{eq:trigger0} and the fact that $|y(t)|\leq c$, it can be shown that for any two consecutive time instants \cite{lazar2003time,lazar2004time}:
\begin{equation}
 \frac{\kappa\delta}{b+c} \leq t_{n+1}-t_n \leq \frac{\kappa\delta}{b-c}.
   \label{eq:consecutive_time}
\end{equation}
\subsection{FRI Signal Recovery}
Consider an FRI signal of the form
\begin{equation}
    x(t) = \sum_{\ell=1}^L a_{\ell} h(t-\tau_{\ell}),
    \label{eq:fri00}
\end{equation}
where the FRI parameters $\{(a_{\ell},\tau_{\ell})|\tau_{\ell} \in (0, T],a_{\ell}\in\mathbb{R} \}_{\ell=1}^L$ are the unknown amplitudes and delays. We assume that the pulse $h(t)\in L^2(\mathbb{R})$, and the number of FRI pulses $L$ are known. 
Since the analysis of recovering aperiodic FRI signals using IF-TEM measurements is similar to that of recovering periodic FRI signals \cite{naaman2021fri}, in this paper we will concentrate on the scenario of recovering $T$-periodic FRI signals.

Consider a $T$-periodic FRI signal, resulted from the linear combination of delayed versions of a prototype pulse $h(t)\in L^2(\mathbb{R})$, of the form
\begin{equation}
    x(t) = \sum_{p\in\mathbb{Z}}\sum_{\ell=1}^L a_{\ell} h(t-\tau_{\ell}-p T),
    \label{eq:fri}
\end{equation}
where the FRI parameters $\{(a_{\ell},\tau_{\ell})|\tau_{\ell} \in (0, T],a_{\ell}\in\mathbb{R} \}_{\ell=1}^L$ correspond to the unknown amplitudes and delays.
The rate of innovation of $x(t)$ is $\frac{2L}{T}$ and hence, $2L$ measurements are sufficient for perfect recovery \cite{vetterli2002sampling, eldar2015sampling}.
\begin{figure}
    \centering
    \includegraphics[width=0.5\textwidth]{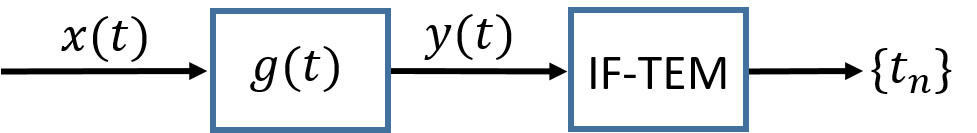}
    \caption{Sampling setup IF-TEM: Continuous-time signal $x(t)$ is filtered through a sampling kernel $g(t)$ and then sampled by using an IF-TEM to generate time instances $\{t_n\}$.}
    \label{fig:IAF}
\end{figure}

Since $x(t)$ is $T$-periodic, it has a Fourier series representation
\begin{equation}
    x(t) = \sum_{k\in\mathbb{Z}}\hat{x}[k]e^{jk\omega_0 t},
    \label{eq:x_initial}
\end{equation}
where $\omega_0=\frac{2\pi}{T}$. The Fourier-series coefficients (FSCs) are given by 
\begin{equation}
 \hat{x}[k] = \frac{1}{T}\hat{h}(k\omega_0)\sum_{\ell=1}^{L}a_{\ell}e^{-j k\omega_0\tau_{\ell}},
 \label{eq:x_fourier}
\end{equation}
where $\mathcal{K}$ is a set of integers, $\hat{h}(\omega)$ is the continuous-time Fourier transform of $h(t)$
 \cite{eldar2015sampling}. It is assumed that $\hat{h}(k\omega_0)\neq 0$ for $k\in\mathcal{K}$.

It was shown in \cite{vetterli2002sampling}, that the parameters $\{a_{\ell}, \tau_{\ell}\}_{\ell=1}^L$ can be uniquely computed from $2L$ samples of the FSCs $\hat{x}[k]$ using spectral analysis methods, such as the annihilating filter (AF) \cite{eldar2015sampling}. Thus, FRI signal reconstruction is reduced to the problem of uniquely determining the desired number of FSCs from the signal measurements. 
\subsection{Kernel and Sub-Nyquist Sampling}
\label{sec:kernel}
A crucial component of an FRI sampling architecture is the sampling kernel. Generally, sampling kernels with compact support are preferable from a hardware implementation perspective. We consider IF-TEM sampling and recovery with a compactly supported sum-of-sincs (SoS) kernel for FRI signals. Consider an SoS kernel 
 generated by
\begin{equation}
\hat{g}(\omega) = \sum_{k\in \mathcal{K}} \text{sinc}\left( \frac{\omega}{\omega_0}-k \right).
\label{eq:SoStime}
\end{equation}
Based on the robust sampling kernel presented in \cite{naaman2021fri}, and to maintain the real-valued nature of the filter response and output, we select $\mathcal{K}$ as 
\begin{equation}
    \mathcal{K} = \{-K,\cdots,-1,1,\cdots,K \}, \,\,\,\,\,\text{where}\,\,\,\,\,K\geq 2L.
    \label{eq:KK}
\end{equation}
The sampling kernel resilience is a result of selecting a support set $\mathcal{K}$ that is symmetric about zero but does not include zero.
The filtered signal $y(t) = (x*g)(t)$ is given as
\begin{equation}
    y(t) =\sum_{k\in \mathcal{K}}
    \hat{x}[k]\hat{g}(k\omega_0)  e^{jk\omega_0 t}=\sum_{k\in \mathcal{K}}
    \hat{x}[k] e^{jk\omega_0 t}.
    \label{eq:yt_by_x22}
\end{equation}
In this case, the forward model or the relation between $y_n$'s and the desired FSCs is given by 
\begin{equation}
    \begin{split}
    y_n =\sum_{k\in \mathcal{K}} \frac{\hat{x}[k]}{jk\omega_0}\left( e^{jk\omega_0 t_{n+1}}- e^{jk\omega_0 t_{n}} \right).
    \end{split}
    \label{eq:yx_rel}
\end{equation}
It was shown in \cite{naaman2021fri}, that $y(t)$ is bounded provided that $\max\{a_{\ell}|a_{\ell}\in\mathbb{R} \}_{\ell=1}^L < \infty$ and the pulse $h(t)$ is absolutely integrable. 

To extract the FSCs from \eqref{eq:yx_rel}, let $\mathbf{y}=$  $[\int_{t_1}^{t_2}y(t)dt, \int_{t_2}^{t_3}y(t)dt, \cdots, \int_{t_{N-1}}^{t_N}y(t)dt]^{\top}$, where $N$ is the number of time instants in the interval $T$.
The measurements $\mathbf{y}$ and the FSCs
\begin{figure*}[t]
 \begin{align}
  \mathbf{B}=     \begin{bmatrix}
 e^{-jK\omega_0t_2}-e^{-jK\omega_0t_1}& \cdots & e^{-j\omega_0t_2}-e^{-j\omega_0t_1}&e^{j\omega_0t_2}-e^{j\omega_0t_1}  & \cdots & e^{jK\omega_0t_2}-e^{jK\omega_0t_1}\\       
 e^{-jK\omega_0t_3}-e^{-jK\omega_0t_2}& \cdots & e^{-j\omega_0t_3}-e^{-j\omega_0t_2}&e^{j\omega_0t_3}-e^{j\omega_0t_2}  & \cdots & e^{jK\omega_0t_3}-e^{jK\omega_0t_2}\\ 
 \vdots&  & \vdots & \vdots&  & \vdots\\
e^{-jK\omega_0t_N}-e^{-jK\omega_0t_{N-1}}& \cdots &e^{-j\omega_0t_N}-e^{-j\omega_0t_{N-1}}&e^{j\omega_0t_N}-e^{j\omega_0t_{N-1}}  & \cdots & e^{jK\omega_0t_N}-e^{jK\omega_0t_{N-1}}\\  
\end{bmatrix}.
\label{eq:Bmat}
  \end{align}
\end{figure*}
\begin{align}
\mathbf{\hat{x}} = \left[-\frac{\hat{x}[-K]}{jK\omega_0}, \cdots, -\frac{\hat{x}[-1]}{j\omega_0},\frac{\hat{x}[1]}{j\omega_0},\cdots, 
\frac{\hat{x}[K]}{jK\omega_0}\right]^{\top}
\label{eq:xhatProof2}
\end{align}
are related as 
\begin{equation}
    \mathbf{y}=\mathbf{B}\mathbf{\hat{x}},
    \label{eq:mesRel2}
\end{equation}
where $\mathbf{B}$ is given in \eqref{eq:Bmat}. It was shown in \cite{naaman2021fri}, that the matrix $\mathbf{B}$ has full column rank and is uniquely left invertible. Then the Fourier coefficients vector can be computed as
\begin{equation}
    \mathbf{\hat{x}} = \mathbf{B}^\dagger \mathbf{y},
    \label{eq: Moore-Penrose_inverse}
\end{equation}
where $\mathbf{B}^\dagger$ denotes the Moore-Penrose inverse.
Prefect reconstruction is established by \cite{naaman2021fri} when $N \geq 4L+2$ and $|\mathcal{K}|\geq 2L$, as summarized in the following theorem:
\begin{theorem}[Section \RN{3}.D in \cite{naaman2021fri}]
\label{theorem:FRI}
Let $x(t)$ be a $T$-periodic FRI signal of the following form
\begin{align*}
      x(t) = \sum_{p\in\mathbb{Z}}\sum_{\ell=1}^L a_{\ell} h(t-\tau_{\ell}-p T),
\end{align*}
where the pulse $h(t)\in L^2(\mathbb{R})$, and the number of FRI pulses $L$ are known. Consider the sampling mechanism shown in Fig. ~\ref{fig:IAF}. Let the sampling kernel $g(t)$ satisfy
\begin{align*}
    \hat{g}(k\omega_0) =       
    \begin{cases}
      1 & \text{if $k\in \mathcal{K}=\{-K,\cdots,-1,1,\cdots,K \}$},\\
      0 & \text{otherwise},
    \end{cases}
    \end{align*}
 and $\max\limits_t|(h*g)(t)|<\infty$. The filtered signal $y(t) = (x*g)(t)$. Suppose the IF-TEM
parameters $\{b, \kappa, \delta\}$ are chosen such that $b > c$ where $c=\max\limits_t |y(t)|$ , and
\begin{equation}           
\frac{b-c}{\kappa\delta} \geq \frac{2K+2}{T}. \label{eq:sample_bound}
\end{equation} 
Then the parameters $\{a_{\ell}, \tau_{\ell}\}_{\ell=1}^L$ can be perfectly recovered from the IF-TEM outputs if
\begin{enumerate}
    \item $K \geq 2L$ when $\{t_\ell\}_{\ell=1}^L$ are off-grid.
    \item $K \geq L$ when $\{t_\ell\}_{\ell=1}^L$ are on-grid.
\end{enumerate}
\end{theorem}

In practice, our IF-TEM HW circuit introduces noise into the signal, which causes the time occurrences $t_n$ to be perturbed. As a result, unstable recovery occurred when the aforementioned algorithm was utilized in the process of reconstructing the data from the hardware measurements (see Section \ref{sec:REC} for more details). 
Therefore, a reconstruction strategy that is more robust to noise is required.
\subsection{Problem Formulation}

Consider a $T$-periodic FRI signal of the form of \eqref{eq:fri} and a sampling mechanism as shown in Fig. \ref{fig:IAF}. The signal $x(t)$ is passed through the sampling kernel $g(t)$ as defined in \eqref{eq:SoStime}, and the resulting signal $y(t)$ is sampled using an IF-TEM.
Both the time encodings$\{t_n\}_{n=1}^N$ and the amplitude measurements $\{y_n\}_{n=1}^N$ correspond to a discrete representation of $y(t) = (x*g)(t)$. In other words, $\{t_n\}$ encodes information of the FRI signal.
As our objective is to design robust hardware, the FRI parameters $\{a_{\ell},\tau_{\ell}\}_{\ell=1}^L$ need to be accurately estimated from the IF-TEM firings.
To this end, together with the hardware implementation, a robust recovery algorithm is needed. 
In the following section, we first introduce our robust recovery mechanism that perfectly recovers the Fourier series coefficients $\{\hat{x}[k]\}_{k \in \mathcal{K}}$ from IF-TEM observations in the absence of noise with as few as $4L+2$ spikes inside an interval $T$. Then, we illustrate the resilience of our method in the presence of noise and demonstrate that it outperforms the one proposed in \cite{naaman2021fri}. In Section \ref{sec:HW_spec}, we discuss our hardware prototype realizations.


\section{Robust Sub-Nyquist sampling and Reconstruction of FRI Signals from IF-TEM }
\label{sec:REC}

The IF-TEM circuit introduces noise into the signal, which perturbs the time instances $\{t_n\}$. Even in the absence of noise, the time instances can only be determined with limited precision. 
The modeled jittered time instances are modeled as
\begin{equation}
    t_n^{\prime} =t_n+ \epsilon_n,
    \label{eq:jitter}
\end{equation}
where $t_n$ are the ideal time instances and $\epsilon_n\overset{\mathrm{iid}}{\sim} \mathcal{U}[-\frac{\sigma}{2},\frac{\sigma}{2}]$ is the noise jitter.
Our experiments on our hardware showed that the noise level $\sigma$ fluctuates between $15-70$ $ms$.
Since the method presented in \cite{naaman2021fri} to reconstruct FRI signals from IF-TEM measurements resulted in an inconsistent recovery using hardware data, a more noise-tolerant reconstruction method is presented next. 

We compare the proposed recovery method with the reconstruction described by \cite{naaman2021fri} in the presence of perturbations to the measured time instances. Both approaches employ a sample kernel lacking the zeroth frequency. While the reconstruction approach proposed by \cite{naaman2021fri} used the forward equation defined in \eqref{eq:yx_rel}, our new algorithm is based on an alternative formulation presented in \eqref{eq:partial} below.
\subsection{Robust Reconstruction}
This section presents a method for determining the Fourier coefficients of the FRI signal that is more robust and improves recovery.
The recovery approach described in \cite{naaman2021fri}, and discussed in the previous section, is based on computing the FSCs $\mathbf{\hat{x}}$ of the FRI signal $x(t)$ using \eqref{eq: Moore-Penrose_inverse}. In the case of noise, this leads to a perturbation in the measurements $y_n$ as well as the matrix  $\mathbf{B}$ defined in \eqref{eq:yx_rel} and \eqref{eq:Bmat}, respectively. In this case, while computing the FSCs, the stability of $\mathbf{B}$, which is measured by the condition number of the matrix, impacts the results. Next, we show that by utilizing a partial summation of $y_n$, perfect recovery is achieved similarly to Theorem 1. In the noisy scenario, the resulting method is more robust. As we show below, when employing a partial summation of $y_n$, we end up with a recovery problem similar to \eqref{eq: Moore-Penrose_inverse} but with the matrix $\mathbf{A}$ defined in \eqref{eq:matrixB} replacing $\mathbf{B}$. This matrix has a better condition number than $\mathbf{B}$. 

To gain intuition as to why this is the case, we demonstrate that for every $k\in\mathcal{K}$, utilizing the partial summation for the measurements reduces the noise in each element of $\mathbf{A}$ by half compared to its corresponding element in $\mathbf{B}$. This result is summarized in the following lemma.
\begin{lemma}
\label{lemma1}
Let 
$[\mathbf{B}]_{nk} = e^{jk\omega_0 t_{n+1}}- e^{jk\omega_0 t_{n}}$ be the entries of matrix $\mathbf{B}$, where $n=1,\cdots,N-1$, $k\in\mathcal{K}$. 
Let $[\mathbf{A}]_{nk}=e^{jk\omega_0 t_{n+1}}$ be the entries of matrix $\mathbf{A}$, where $n=1,\cdots,N-1$, $k\in\mathcal{K}\cup\{0\}$. 
The jittered time instances are modeled as $t_n^{\prime} =t_n+ \epsilon_n$, where $t_n$ are the ideal time instances and the jitter noise is modeled as $\epsilon_n\overset{\mathrm{iid}}{\sim} \mathcal{U}[-\frac{\sigma}{2},\frac{\sigma}{2}]$, i.i.d. uniformly distributed.
For every $t_n^{\prime}$ and $k\in\mathcal{K}$, 
\begin{equation}
    \text{var}\left([\mathbf{B}]_{nk}\right)=2\text{var}\left([\mathbf{A}]_{nk}\right),
\end{equation}
where $\text{var}$ is the variance.
\end{lemma}
\begin{proof}
By utilizing the fact that $t_n^{\prime} = t_n + \epsilon_n$, and using \eqref{eq:Bmat} and \eqref{eq:matrixB}, it follows that,
\begin{equation}
    \begin{split}
        \text{var}\left([\mathbf{B}]_{nk}\right)&= \text{var} \left(e^{jk\omega_0 (t_{n+1}+\epsilon_{n+1})}- e^{jk\omega_0 (t_{n}+\epsilon_{n})}\right)\\&= \lvert e^{jk\omega_0 t_{n+1}}\rvert ^2\,\text{var} \left(e^{jk\omega_0 \epsilon_{n+1}}\right)
        \\&+\lvert e^{jk\omega_0 t_{n}}\rvert ^2\, \text{var} \left(e^{jk\omega_0 \epsilon_{n}}\right)\\&=2\text{var}\left(e^{jk\omega_0 \epsilon_{n}}\right)=2 \text{var}\left([\mathbf{A}]_{nk}\right),
    \end{split}
\end{equation}
establishing the lemma.
\end{proof}

It can be intuitively inferred that by utilizing the partial summation, the noise in each element of $[\mathbf{A}]_{nk}$ becomes smaller than the corresponding noise in $[\mathbf{B}]_{nk}$. Consequently, the matrix $\mathbf{A}$ has a better condition number than $\mathbf{B}$. This can be explained by the fact that the condition number of a matrix is a measure of the sensitivity of the matrix to small perturbations in its elements, and a smaller condition number indicates that the matrix is less sensitive to such perturbations. Therefore, by reducing the noise in the elements of $\mathbf{A}$ using the partial summation, we can improve its condition number.

In the next step, we employ the partial summation of $y_n$ to present a perfect recovery guarantee for FRI signals by using IF-TEM. Instead of recovering the FSCs from $y_n$ through the forward model \eqref{eq:yx_rel} with $\mathcal{K}$ in \eqref{eq:KK}, which defines the relation between $y_n$ and the FSCs ${\hat{x}[k]}$, we propose an alternative model that is based on $z_n$. These are the partial sums of the measurements $y_n$ defined as
\begin{equation}
    z_n=\sum_{i=1}^{n-1} y_i = \sum_{k\in \mathcal{K} } \frac{\hat{x}[k]}{jk\omega_0} \left( e^{jk\omega_0 t_{n}}- e^{jk\omega_0 t_{1}} \right),
    \vspace{.1in}
    \label{eq:partial}
\end{equation}
where $ n = 2,\cdots,N$.
Note that \eqref{eq:partial} can be written as 
\begin{equation}
    z_n= \sum_{k\in \mathcal{K} } \frac{\hat{x}[k]}{jk\omega_0}  e^{jk\omega_0 t_{n}}+ c,
    \vspace{.1in}
    \label{eq:partial2}
\end{equation}
where 
\begin{equation}
     c = - \sum_{k\in \mathcal{K} } \frac{\hat{x}[k]}{jk\omega_0}  e^{jk\omega_0 t_{1}}.
\end{equation}
Let $\mathbf{z} = [z_2,\cdots,z_N]^\mathrm{T}\in\mathbb{R}^{N-1}$ be the vector of partial sums, $ \mathbf{\hat{z}} =
\left[-\frac{\hat{x}[-K]}{jK\omega_0}, \cdots, -\frac{\hat{x}[-1]}{j\omega_0},c,\frac{\hat{x}[1]}{j\omega_0},\cdots, \frac{\hat{x}[K]}{jK\omega_0}\right]^{\top}$ $\in\mathbb{C}^{(2K+1)}$ be the vector of FSCs, with $c$ in the zeroth place, and  $\mathbf{A}\in\mathbb{C}^{(N-1) \times (2K+1)}$ be the matrix defined as
\begin{equation}
    \mathbf{A} = {\begin{bmatrix}e^{-jK\omega_0 t_2}& \cdots1\cdots&{e^{jK\omega_0 t_2}}\\
     e^{-jK\omega_0 t_3}&\cdots1\cdots&{e^{jK\omega_0 t_3}}\\ 
     \vdots &\ddots&\vdots \\ 
     e^{-jK\omega_0 t_N}& \cdots1\cdots&{e^{jK\omega_0 t_N}}
    \end{bmatrix}}.
    \label{eq:matrixB}
\end{equation}
Then, \eqref{eq:partial2} can be expressed in matrix form as follows:
\begin{equation}
     \mathbf{z} = \mathbf{A}\, \mathbf{\hat{z}}.
     \label{eq:forward_model}
\end{equation}

Since the set of time instants $\{t_n\}_{n=2}^N$ are distinct, and $\mathbf{A}$ is a Vandermonde matrix, it has full column rank provided that $N-1 \geq 2K+1$. This means that the matrix $\mathbf{A}$ has linearly independent columns. Therefore, we can perfectly recover the vector of FSCs $\mathbf{\hat{z}}$ via
\begin{equation}
    \mathbf{\hat{z}} = \mathbf{A}^\dagger\,\mathbf{z},
    \label{eq:fourierZ}
\end{equation}
where $\mathbf{A}^\dagger$ denotes the Moore-Penrose inverse of $\mathbf{A}$. Once we have $\mathbf{\hat{z}}$, the FSCs $\hat{x}[k]$ can be uniquely determined.
Using
\begin{align}
    \hat{z}[k]= \begin{cases}
     \frac{\hat{x}[k]}{j\omega_0k}, & \text{if $k\in\mathcal{K}$ },\\
      -\sum_{k'\in \mathcal{K}}\left( \frac{\hat{x}[k']}{jk'\omega_0} \right)  e^{jk'\omega_0 t_1} & \text{if $k=0$ }.
      \end{cases}
      \vspace{.1in}
      \label{eq:get_x}
\end{align}
The vector of FSCs $\mathbf{\hat{z}}$ and the vector of FSCs $\mathbf{\hat{x}}$ are related by:
\begin{equation}
    \mathbf{\hat{x}} = \left[\hat{z}[-K],\cdots,\hat{z}[-1],\hat{z}[1],\cdots,\hat{z}[K]\right]^{\top} \in\mathbb{C}^{2K}.
    \label{eq:zandx}
\end{equation}
This equation allows us to obtain $\mathbf{\hat{x}}$ by selecting the appropriate elements of $\mathbf{\hat{z}}$, which is the vector obtained from the partial sums of the measurements. Note that the resulting vector $\mathbf{\hat{x}}$ has dimensions $2K$, which implies that it only contains the FSCs for positive and negative frequencies.

Using the vector $\mathbf{\hat{z}}$ and \eqref{eq:zandx}, the vector of FSCs $\mathbf{\hat{x}}$ is uniquely determined. This indicates that, in the modified kernel setup, without zero frequency, the set of FSCs ${\hat{x}[k]}$ can be uniquely determined from time encodings if $N-1 \geq 2K+1$. This condition implies that there should be a minimum of $2K+2$ firing instants within an interval $T$. For an IF-TEM, the minimum firing rate is given as $\frac{b-c}{\kappa\delta}$ (as shown in \cite{naaman2021fri}). Hence, for uniqueness recovery, the IF-TEM parameters must satisfy the inequality $ \frac{b-c}{\kappa\delta} \geq \frac{2K+2}{T}$ (see \cite{naaman2021fri} for details).

A reconstruction algorithm to compute the FRI parameters from IF-TEM firings is presented in Algorithm~\ref{alg:algorithm_grid}. Compared to the technique presented in \cite{naaman2021fri}, our method requires the same number of FSCs in the absence of noise. However, in the presence of noise, as is typically the case in real-world hardware, the proposed approach yields a lower error for the same number of measurements.

\begin{algorithm}[h]
\begin{algorithmic}[1]
\caption{ Reconstruction of a $T$-periodic FRI signal.}
\label{alg:algorithm_grid}
\Statex \textbf{Input:} $N\geq 2K+2$ spike times $\{t_n\}_{n=1}^N$ in a period $T$.
\State Let $n \leftarrow 1$
        \While {$n\leq N-1$}  
             \State   Compute  
            $y_n = -b(t_{n+1}-t_n)+\kappa\delta$
            \State   Compute  
            $z_{n+1}=\sum_{i=1}^{n} y_i$
            \State $n:=n+1$.
\EndWhile
    \State Compute the vector $\mathbf{\hat{z}} = \mathbf{A}^\dagger\,\mathbf{z}$, where $\mathbf{A}$ is defined in \eqref{eq:matrixB}
    \State Compute the Fourier coefficients vector $\mathbf{\hat{x}}$ from $\mathbf{\hat{z}}$ using \eqref{eq:zandx}
    \State   Estimate $\{(a_{\ell},\tau_{\ell})\}_{\ell=1}^L$ using a spectral analysis method for $K \geq 2L$.
\Statex \textbf{Output:} $\{(a_{\ell},\tau_{\ell})\}_{\ell=1}^L$.
\end{algorithmic}
\end{algorithm}
\begin{figure}[t!]
		\centering
		\includegraphics[width=0.5\textwidth]{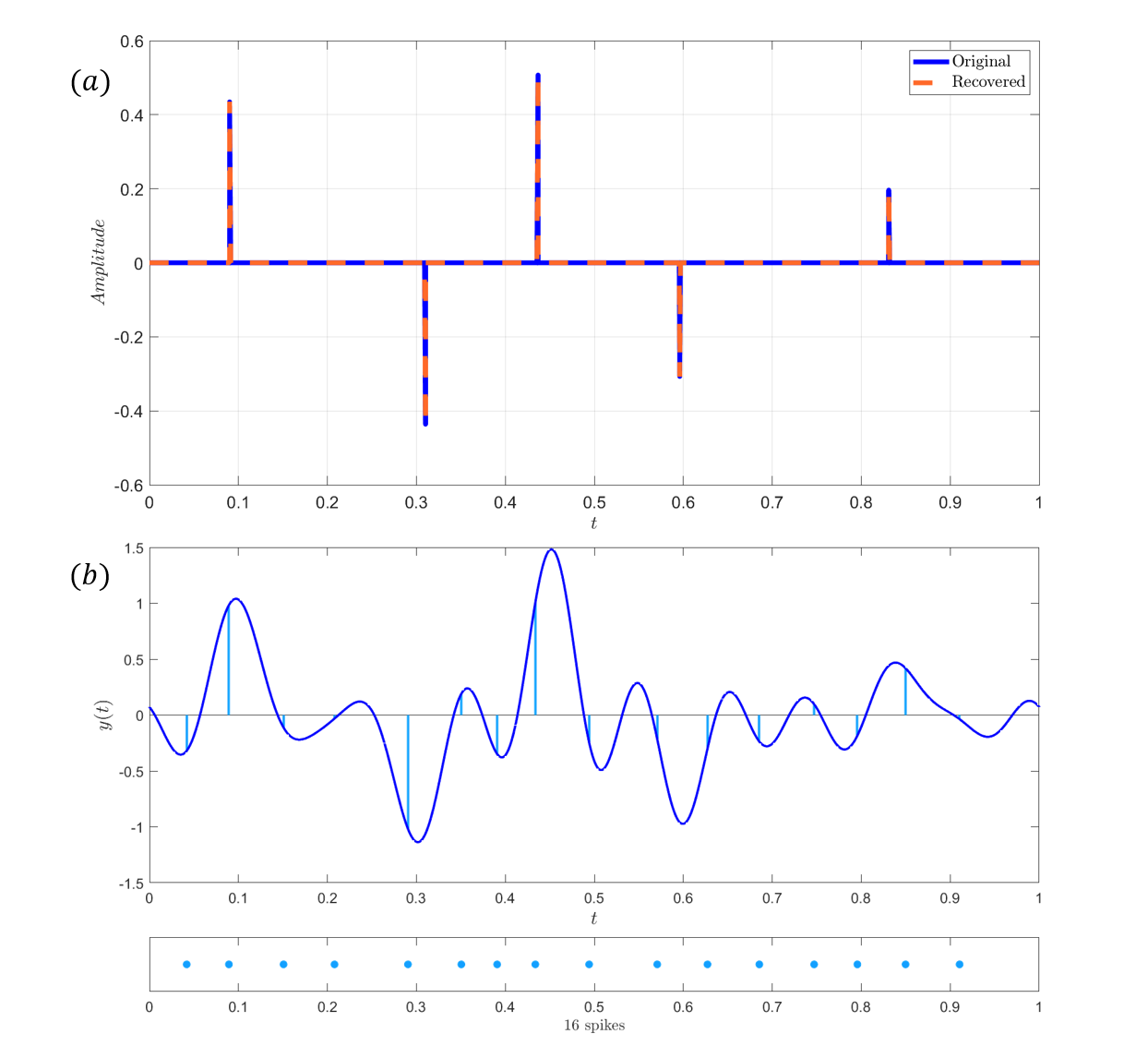}
		\caption{Perfect reconstruction of an FRI signal from IF-TEM measurements with the modified sampling kernel. (a): the input signal and its reconstruction for $L=5$. (b): the filtered signal $y(t)$ and the time instants $t_n$. }
		\label{fig:deltas}
	\end{figure}
\subsection{Numerical evaluation}
In this section, we provide numerical evidence of the validity of Algorithm 1 through simulations. We then show that our proposed reconstruction technique improves the conditioning of the forward transformation, leading to a significant reconstruction improvement that is necessary for precise recovery using real hardware.
To validate Theorem \ref{theorem:FRI}, we consider $h(t)$ as a Dirac impulse with a time period of $T=1$ second, and $L=5$. The amplitudes are selected randomly over the range $[-1,1]$. The time delays are chosen randomly between $(0,1)$ such that they lie on a grid with a resolution of $0.05$. The input signal $x(t)$ is filtered using an SoS sampling kernel with $\mathcal{K} = {-K,\cdots-1,1,\cdots,K}$, where $K=L$. The filtered output $y(t)$ is sampled using an IF-TEM where the IF-TEM parameters are chosen to satisfy the inequality \eqref{eq:sample_bound}. In this particular case, the IF-TEM sampler resulted in $16$ firing instants in one time period, as shown in Fig. \ref{fig:deltas}(b). As shown in Fig. ~\ref{fig:deltas}(a), we achieve perfect recovery of the FRI signal by using a kernel without zero frequency.

As the IF-TEM circuit produces noise into the signal, which perturbs the time instances $\{t_n\}$, we consider the jittered time instances: 
    $t_n^{\prime} =t_n+ \epsilon_n$,
as defined in \eqref{eq:jitter}.
We compare the proposed recovery method with the reconstruction algorithm described by \cite{naaman2021fri} in the presence of perturbations to the measured time instances. Each approach employs a sample kernel lacking a zeroth frequency. While the reconstruction approach proposed by \cite{naaman2021fri} used the forward method defined in \eqref{eq:yx_rel}, our method utilized a different forward method defined in \eqref{eq:partial}.

Using the forward operators or matrices $\mathbf{A}$ and $\mathbf{B}$ (see \eqref{eq:mesRel2} and \eqref{eq:forward_model}), the FSCs are recovered in each of the previously discussed methods. The matrices $\mathbf{A}$ and $\mathbf{B}$ are functions of the measured time instants and sampling kernel. In Fig. \ref{fig:condnum}, the condition numbers of matrices with the same number of $4L+2$ perturbed firing instants are compared as a function of the number of FRI signals $L$. To this aim, 5000 random sets of monotonic sequences $\{t_n\in  [0,T)\}_{n=1}^N$ were generated. As depicted in Fig. \ref{fig:condnum}, the condition number of matrix $\mathbf{A}$ is smaller than that of matrix $\mathbf{B}$. This demonstrates that our reconstruction algorithm enhances stability and noise resilience.

\begin{figure}[t!]
    \centering
    \includegraphics[width=0.5\textwidth]{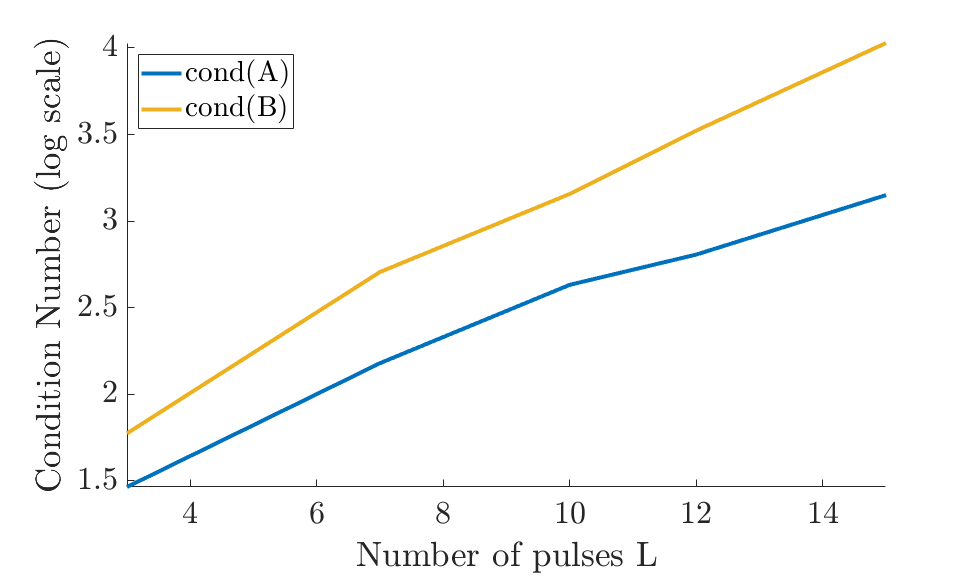}
    \caption{Average condition number of matrices A and B as a function of the number of FRI pulses L.}
    \label{fig:condnum}
\end{figure}

We evaluate and compare the relative mean square error (MSE) in the estimation of time delays performance for the reconstruction accuracy of the presented algorithm with the one suggested by \cite{naaman2021fri}, where the MSE is given by 
\begin{align}
    \text{MSE} = 10 \log\left(\sum_{\ell=1}^L (\tau_\ell-\hat{\tau}_\ell)^2\right),
    \vspace{.1in}
    \label{eq:mse_td}
\end{align}
and $\hat{\tau}_\ell$ is the estimated time delay.
Specifically, we consider the signal $x(t)$ as in \eqref{eq:fri}, with period $T = 1$ second consisting of $L=3$ pulses with $h(t)$ a third-order cubic B-spline. The off-grid time-delays $\{\tau_{\ell}\}_{\ell=1}^3$ and amplitudes $\{a_{\ell}\}_{\ell=1}^3$ are generated at random over intervals $(0, T]$ and $[1, 5]$, respectively. The IF-TEM parameters are $\kappa=1$, $b= 2.5c$ where $c=\max\limits_t |y(t)|$, and $\delta$ is chosen to satisfy \eqref{eq:sample_bound}. 
We consider a sum-of-sincs kernel with $\mathcal{K} =\{-K, \cdots, -1,1 \cdots, K\}$ for the calculation of the Fourier coefficients $\hat{x}[k]$. The time instances $\{t_n\}$ were perturbed as $t_n^{\prime} = t_n + \epsilon_n$ where $t_n$ is the actual time encoding and $\epsilon_n$ is a random variable uniformly distributed over $[-\sigma/2, \sigma/2]$.
We use an annihilating filter with Cadzow denoising to estimate the time-delays in the presence of noise. Since Cadzow denoising requires more than $2L$ consecutive samples of FSCs, we consider $K \geq 2L+1$ while excluding the zero. Based on the fact that the proposed recovery where $0\notin\mathcal{K}$ estimates $\{\hat{x}[k]\}_{k=-K}^{-1}$ and $\{\hat{x}[k]\}_{k=1}^{K}$, we apply Cadzow denoising on each of these sequences independently and then apply block annihilation \cite{block_af} to determine the time-delays jointly. 

The MSEs in the estimation of time-delays for different numbers of FSCs and perturbation levels are shown in Fig. \ref{fig:TE_compare_0}. We used 500 independent noise and FRI signal realizations to compute each MSE value. In Fig. \ref{fig:TE_compare_0}(a) and (b) we show MSEs for \cite{naaman2021fri} and Algorithm 1, both without zero in the sampling kernel, for $K =2L+1$ up to $5L$. We observe that comparing the approaches, we note a gain of up to $10$ dB. Since perturbation in the time encoding is also equivalent to quantization noise, a lower MSE indicates that our proposed approach can operate at lower bits compared to \cite{naaman2021fri}. 

\begin{figure}[!t]
\begin{center}
\begin{tabular}{c c}
\subfigure[Without zero approach in \cite{naaman2021fri}]{\includegraphics[width = 3in]{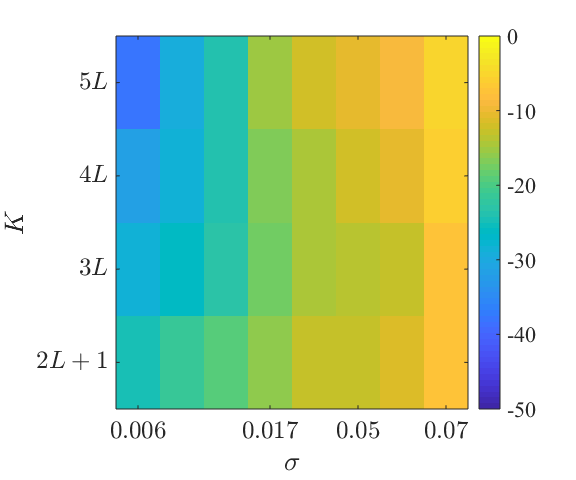}}&\\
\hspace{-.1in}\subfigure[Without zero Algorithm 1]{\includegraphics[width = 3in]{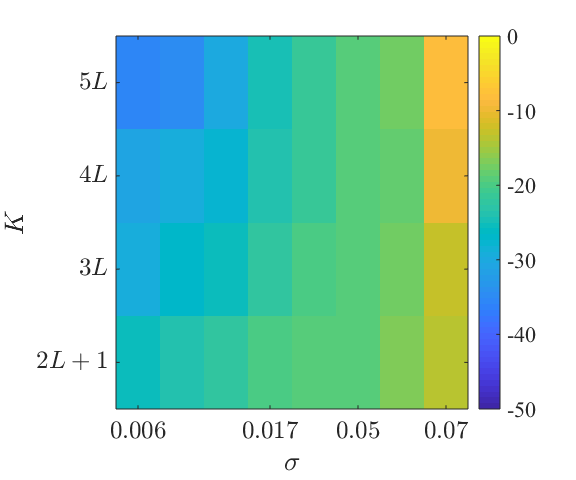}}\vspace{-.1in}
\end{tabular} 
\end{center}
\caption{A comparison of \cite{naaman2021fri} and Algorithm 1 for off-grid time delays with perturbation in the time encodings: our method has lower error compared to \cite{naaman2021fri}.}
\label{fig:TE_compare_0}
\end{figure}


\section{Analog board and Hardware challenges}
\label{sec:HW_spec}
In this section, we will describe the specifications of our FRI-TEM hardware prototype.
\begin{figure}[t]
		\centering
		\includegraphics[width= 0.5\textwidth]{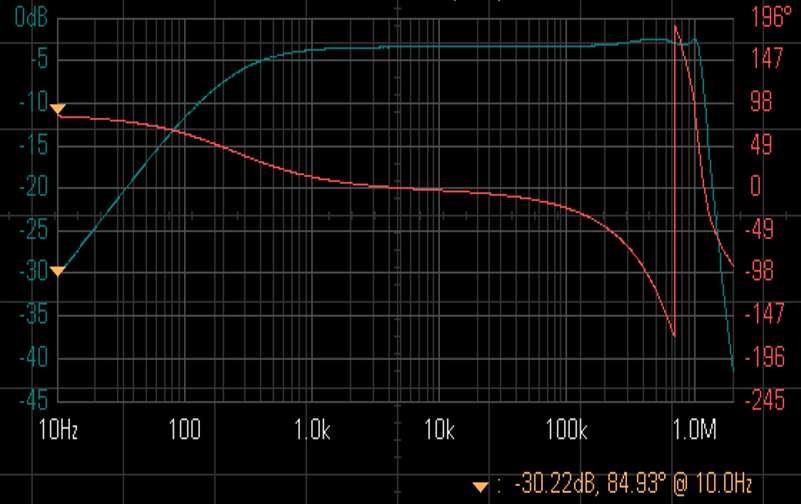}
		\caption{A 1MHz filter Bode plot. The sampling kernel removes the zeroth frequency. The magnitude (in blue) and phase (in red) plotted on a logarithmic frequency scale.}
		\label{fig:1MHz_Bode.png}
\end{figure}
\begin{figure*}[t!]
    \centering
    \includegraphics[width=\textwidth]{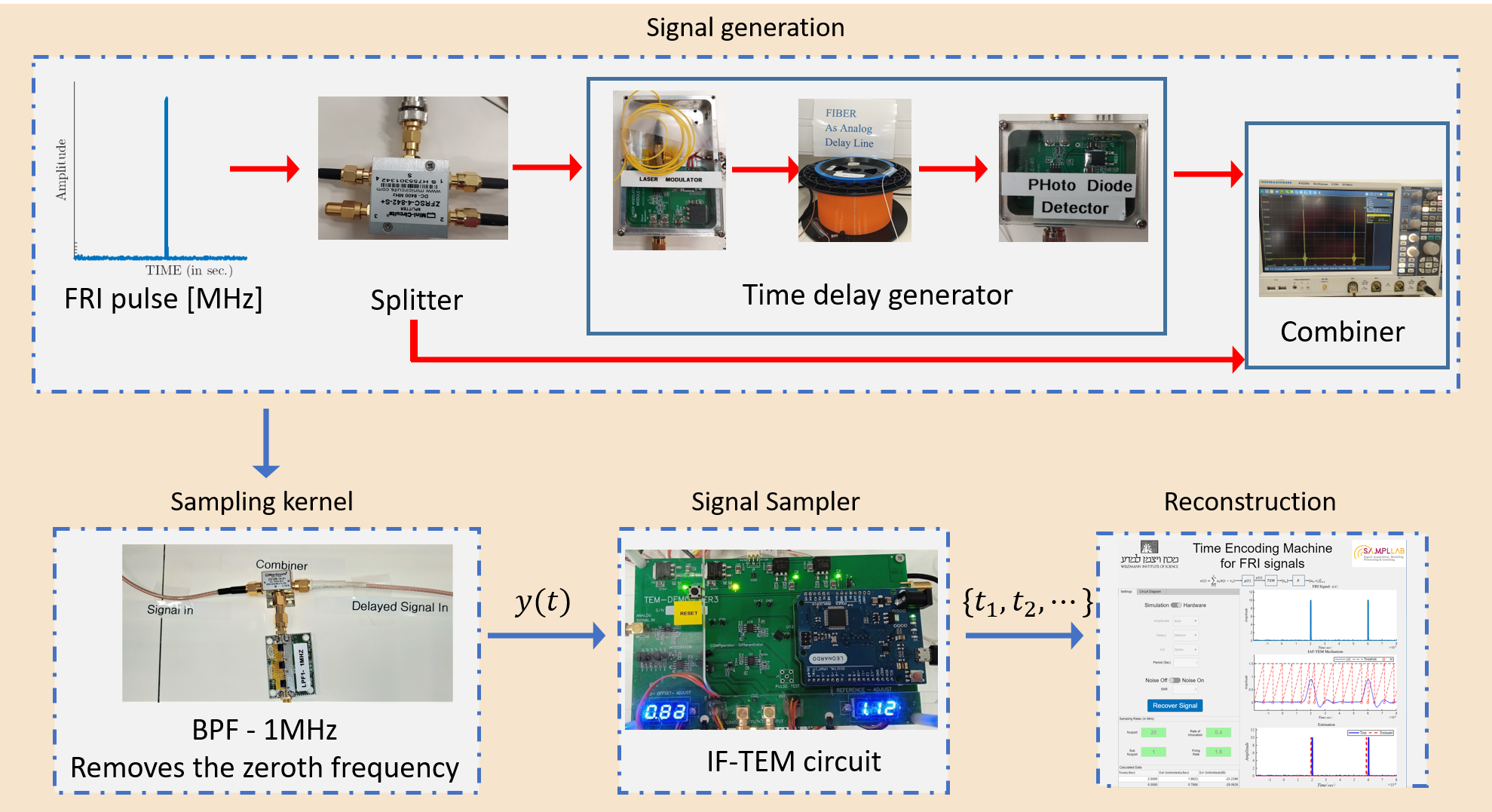}
    \caption{The FRI-TEM hardware prototype contains a signal generator, a sampling kernel and an IF-TEM sampler. The signal generator consists of a delay path of atleast 4$\mu s$, which is built from a modulator optic fiber that ends in a Photo-Diode detector. Then, a combiner receives the  original signal (single one or more) from the generator and the delayed path to create the FRI signal. The generated signal is passed into a BPF of 1$MHz$ which removes the zero frequency. Finally, the resulting signal is sampled by an IF-TEM sampler board. Based on Algorithm 1, the FRI signal is recovered.}
    \label{fig:IAF00}
\end{figure*}
\begin{figure*}[b!]
    \centering
    \includegraphics[width=1.05\textwidth]{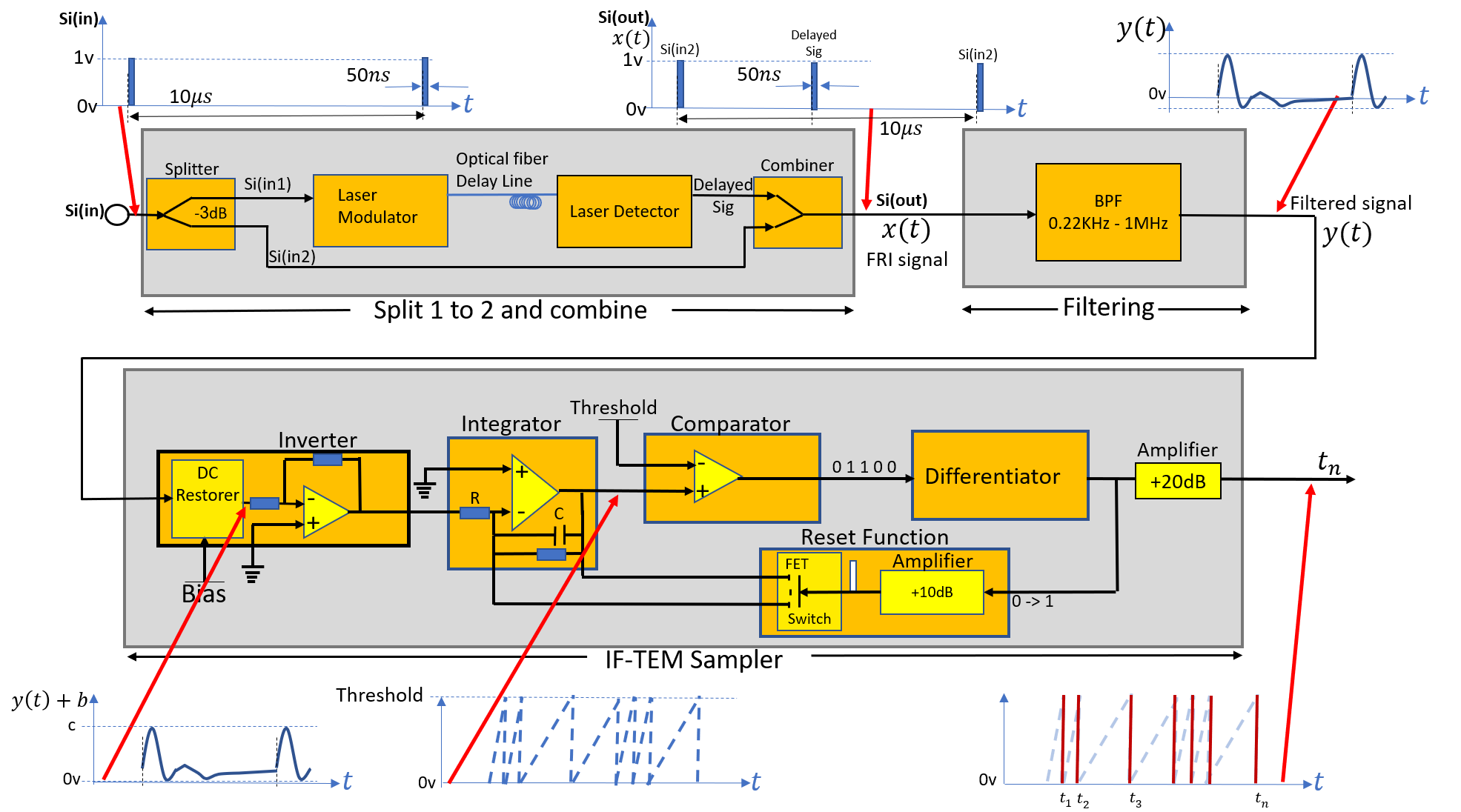}
    \caption{Block diagram of the analog board.}
    \label{fig:replace}
\end{figure*}




\subsection{FRI-TEM Analog Board}
We begin by discussing the key components of the FRI-TEM hardware implementation, as well as various circuit design considerations.
As shown in Fig. \ref{fig:IAF00} and \ref{fig:replace}, the analog board comprises three sequential stages: the generation of an FRI signal, band-pass filtering, and an IF-TEM sampler.

The FRI signal generator uses an analog approach, which is known for its low digital noise and ability to accurately simulate real-world applications such as radar and ultrasound \cite{newberg1990long}.
The process of signal production involves several components working together to generate and process a signal.
One possible configuration for a signal generator is to use a scope, a splitter, an analog delay generator, and a passive radio frequency (RF) combiner. 
The scope generates an FRI pulse ($10-500$ns wide), that is transmitted through the splitter. 
The splitter receives the pulse and sends it to both the delay generator and the combiner (see Fig. \ref{fig:IAF00}).
The delay generator is comprised of a fiber optic cable, a photo-diode encoder, and a photo-diode detector. Encoding the signal with the photo-diode encoder is the initial step of the delay generator. The signal then travels through the fiber optic line, causing a delay of at least $4 \mu s$. The significance of the fiber optic delay implementation originates from its well-known benefits, such as the introduction of low digital noise, which more accurately simulates practical applications. In order to decode the delayed input signal, a photo-diode detector is used to transform the signal to an analog signal with the same frequency as the original FRI input pulse. In Fig. \ref{fig:IAF00} for instance, the FRI signal $x(t)$ \eqref{eq:fri} consists of two 20MHz pulses separated by a relative delay of 4$\mu s$. The output of the combiner, $x(t),$ is then sent as input to the sampling kernel.

\begin{figure*}[h!]
		\centering
		\includegraphics[width=\textwidth]{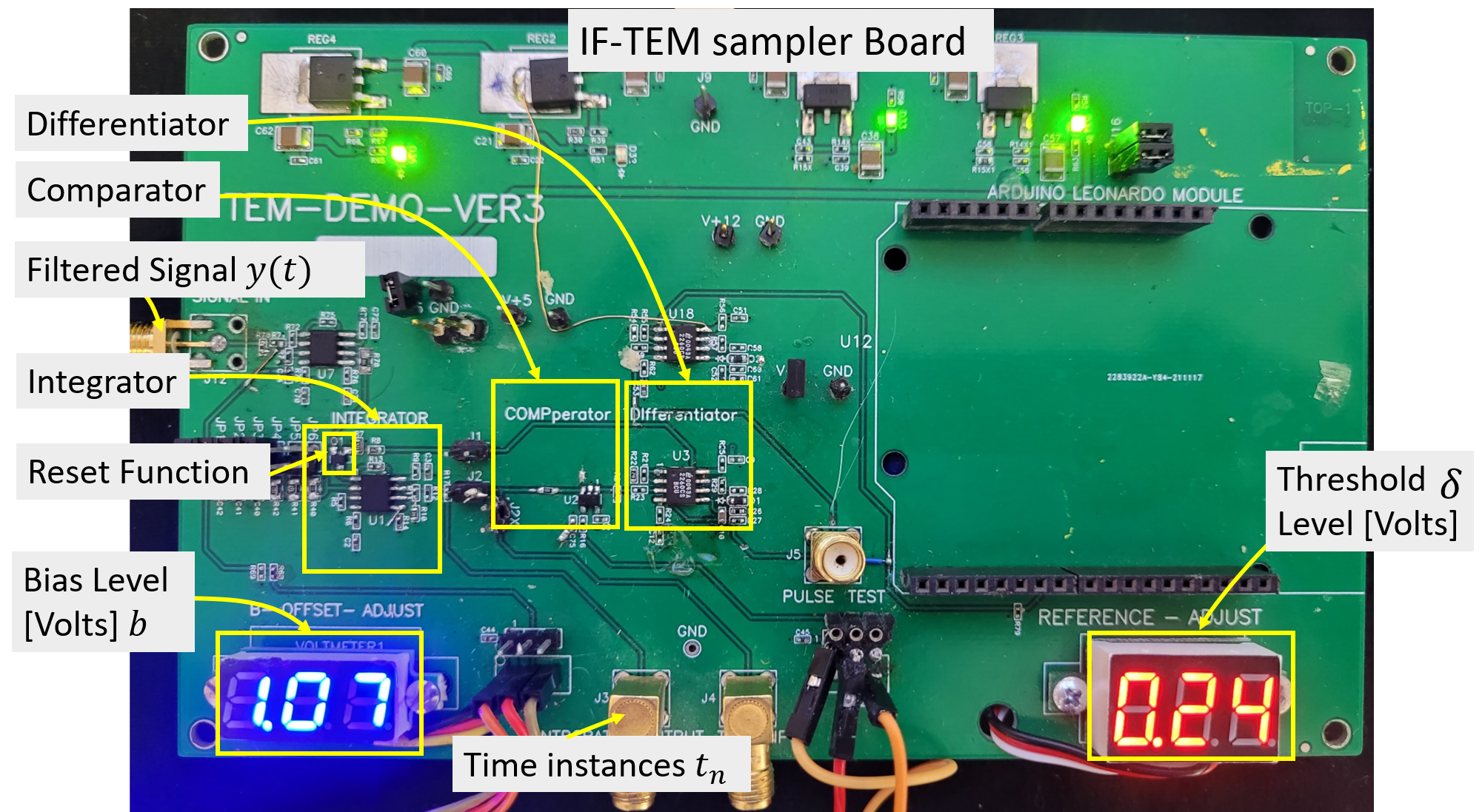}
		\caption{IF-TEM hardware board.}
		\label{fig:IF-TEM_board}
\end{figure*}
\begin{figure*}[t!]
    \centering
    \includegraphics[width=\textwidth]{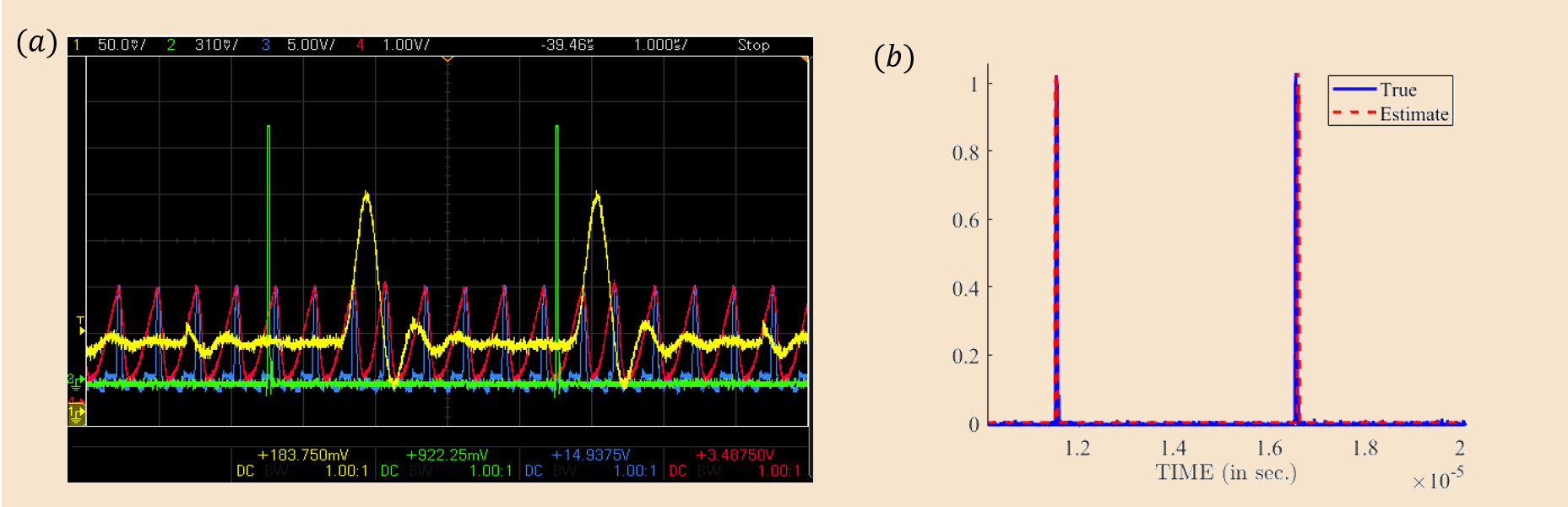}
    \caption{(a). FRI input signal $x(t)$ (green), BPF output $y(t)$ (yellow), and the IF-TEM output resulting in 19 samples (blue). (b). sampling and reconstruction using IF-TEM hardware: the input signal $x(t)$ (blue) and its reconstruction (red).}
    \label{fig:IAF0222}
\end{figure*}
\begin{figure*}[t!]
    \centering
    \includegraphics[width=\textwidth]{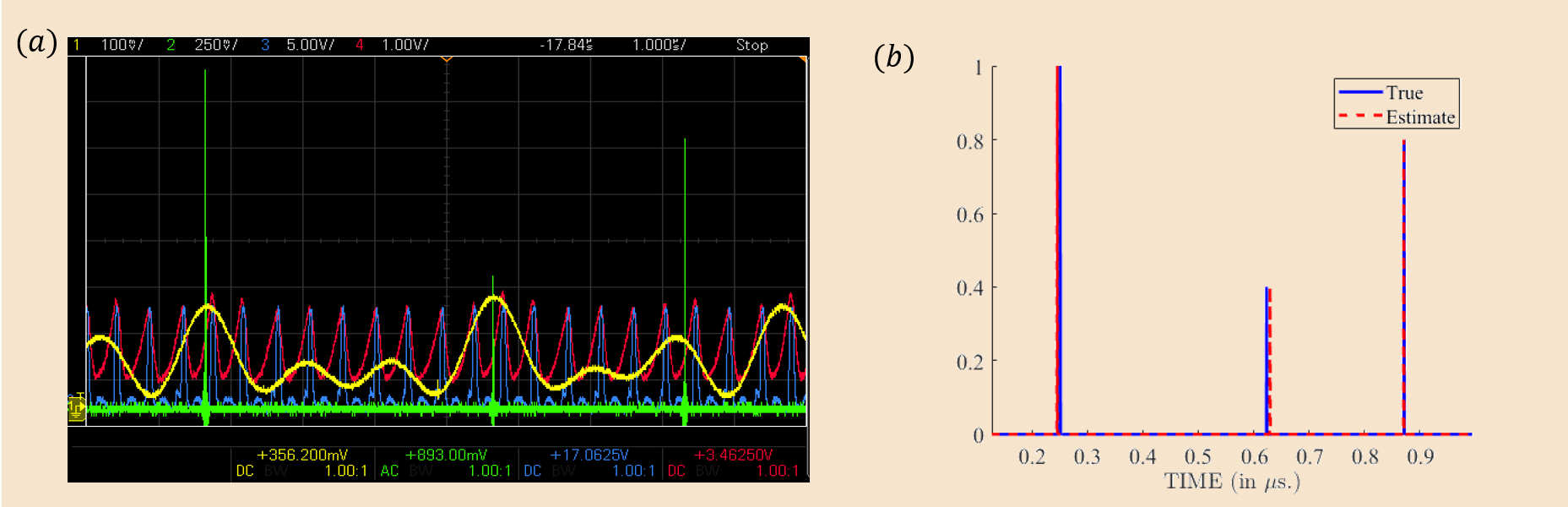}
    \caption{(a). FRI input signal $x(t)$ (green), BPF output $y(t)$ (yellow), and the IF-TEM output resulting in 19 samples (blue). (b). sampling and reconstruction using IF-TEM hardware: the input signal $x(t)$ (blue) and its reconstruction (red).}
    \label{fig:IAF0223}
\end{figure*}
\begin{figure*}[t!]
    \centering
    \includegraphics[width=\textwidth]{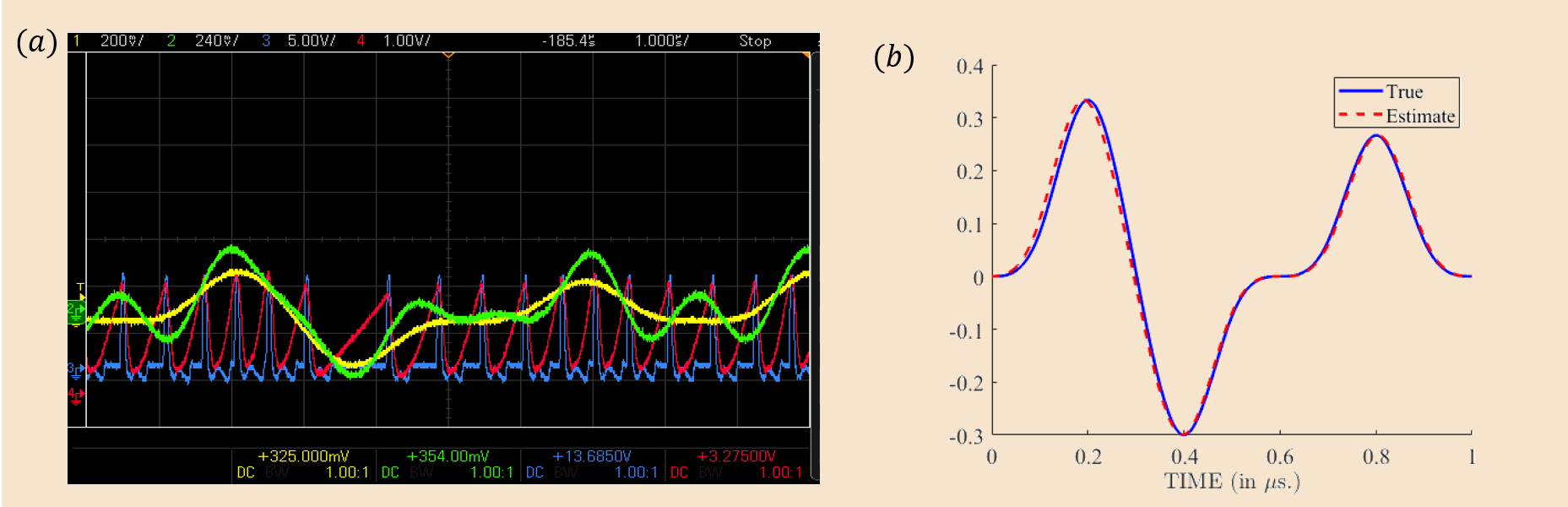}
    \caption{(a). FRI input signal $x(t)$ (yellow), BPF output $y(t)$ (green), and the IF-TEM output resulting in 21 samples (blue). (b). sampling and reconstruction using IF-TEM hardware: the input signal $x(t)$ (blue) and its reconstruction (red).}
    \label{fig:IAF0224}
\end{figure*}
\begin{figure*}[t!]
    \centering
    \includegraphics[width=\textwidth]{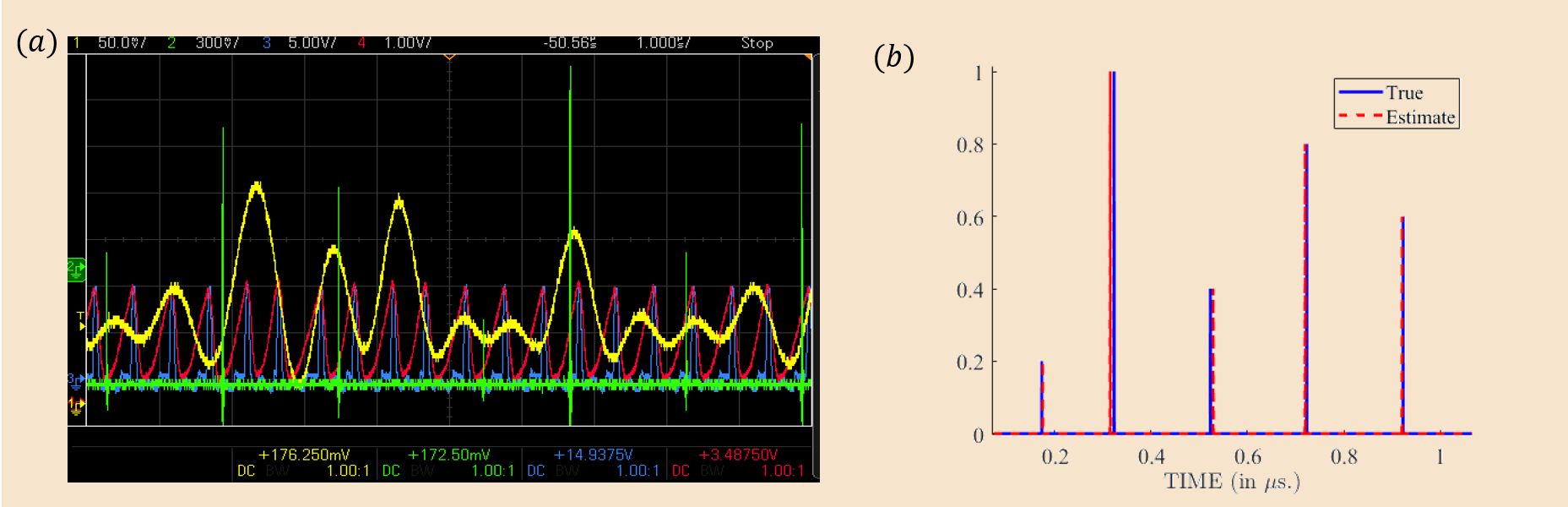}
    \caption{(a). FRI input signal $x(t)$ (green), BPF output $y(t)$ (yellow), and the IF-TEM output resulting in 22 samples (blue). (b). sampling and reconstruction using IF-TEM hardware: the input signal $x(t)$ (blue) and its reconstruction (red).}
    \label{fig:IAF0225}
\end{figure*}
The filter, also known as the sampling kernel, is used to remove the zeroth frequency component of the signal, as shown in Fig. \ref{fig:1MHz_Bode.png}. For example, if the frequency of the signal is 10 Hz, the magnitude of the zeroth frequency component would be $-30$dB. The positioning of the sampling kernel, which is essentially a band-pass filter (BPF), is critical for the sub-Nyquist sampling and reconstruction of FRI signals using an IF-TEM (see Section \ref{sec:REC}). 
In order to accurately recover and analyze two pulses of an FRI signal within a noise-free setting for a short time period such as $T=10 \mu s$, the minimum theoretical sampling rate required is $0.4$MHz \cite{eldar2015sampling,vetterli2002sampling}.
In order to facilitate this fast sampling and reconstruction process, a $1$MHz filter was chosen.
Here, an eight-order $1$MHz BPF is employed, enabling a suitable trade-off between energy usage and reconstruction performance.


The output of the filter, $y(t)$ \eqref{eq:yt_by_x22}, is then transmitted to the IF-TEM sampler. The block diagram of the IF-TEM circuit is shown in Fig. \ref{fig:replace}, and a list of the specific components of the IF-TEM circuit can be found in Table \ref{table:1}. A prototype of the IF-TEM sampler is depicted in Fig. \ref{fig:IF-TEM_board}.

The primary IF-TEM components consist of the bias $b$, integrator, comparator, differentiator, and reset function. To guarantee sufficient samples for reconstruction, we should ensure that the $\delta$ threshold is achieved at least as many times as the desired sample amount. By adding the bias $b$ to the input $y(t)$, the integrator obtains a signal that is always non-negative. In this case, integration over a non-negative signal is a positive function, and the threshold is always attained. For an FRI signal $x(t)$ with $L$ pulses, it can be shown that the sampler input filtered signal $y(t)$ is constrained by \cite{naaman2021fri}
\begin{align}
    |y(t)| \leq c=  
{L\,\,a_{\max} \,\,
    \|g\|_{\infty} \|h\|_{1}},
    \vspace{.1in}
    \label{eq:c}
\end{align}
where $g$ and $h$ are the known filter and pulse shape, respectively. 
Consequently, the bias $b>c$, which is effectively a constant DC voltage, is selected manually using a potentiometer, which is a device that allows the user to adjust the electrical resistance in a circuit by turning a knob. By adjusting the resistance, the user is able to fine-tune the value of the bias to the desired level. It is important to carefully select the appropriate bias value in order to ensure that the IF-TEM system is able to function properly.

The output of the integrator is sent to the comparator, which compares the integrator voltage to a predefined threshold $\delta$.
The threshold is a constant DC voltage that is implemented in our hardware utilizing a potentiometer that is manually regulated and adjustable.
The comparator is responsible for comparing the voltage produced by the integrator to a predefined threshold value. When the integrator voltage reaches or exceeds the threshold, the comparator's output changes. If the comparator's input is below the threshold, it will output a logical value of '0', while if the input is above the threshold, the output will be '1'. In other words, the comparator will produce a sequence of logical '1' values when the integrator voltage hits the threshold. This change in the comparator's output signal indicates that the threshold has been reached and triggers the next stage in the IF-TEM process.

The output of the comparator is sent to the differentiator, which generates a short pulse that activates the fast reset function. This function is responsible for capturing the time instances $t_n$. The reset function consists of an amplifier and a field-effect transistor (FET) that work together to quickly and completely discharge the integrator capacitor.
In greater detail, the FET functions as a switch and is controlled by the pulse produced by the differentiator, which determines the duration of time that the FET is active. This allows the integrator capacitor to be fully discharged. The FET has three terminals: source, gate, and drain. By providing a voltage of "1" to the gate terminal, the FET can modify the conductivity between the drain and source terminals, which allows the current flow to be regulated. This results in a rapid and complete discharge of the integrator capacitor.


\begin{table}[h!]
\centering
\caption{List of Hardware Components}
\label{table:1}
\begin{tabular}{||c c c||} 
 \hline
 Device & Reference & Manufacturer \\ [0.5ex] 
 \hline\hline
 Buffer & AD899 & Mini-Circuits \\
 Integrator &  LT1364 & Analog Devices \\
 Comparator & TLV3201 & Texas Instruments \\
 Differentiator & LT1364 & Analog Devices \\ [1ex]
 \hline
\end{tabular}
\end{table}



\subsection{Circuit challenges}
To implement an IF-TEM circuit, it is necessary to employ an integrator that operates according to \eqref{eq:trigger0}. Specifically, the integrator capacitor must operate in its linear domain, which is continuously charged or rising, as long as the input signal is positive.
Additionally, the IF-TEM thresholding process requires a fast reset mechanism. Therefore, our goal is to develop an integrator and reset function in which the capacitor of the integrator operates in its linear zone and discharges rapidly and completely.
The main challenge in the implementation of the IF-TEM hardware is to design and implement such an IF-TEM integrator capacitor, while supporting a wide range of input FRI signals without circuit modification. 
By utilizing the differentiator and a FET in the reset function, both the entire discharge and rapid discharge of the capacitor are accomplished for a variety of FRI signals.
Next, we provide results from our hardware and compare them to our theoretical results from Theorem \ref{theorem:FRI}.
\section{Hardware Experiments}
\label{sec:HW_experiments}
To determine the potential and feasibility of the development proposed system, we performed experiments on the FRI-TEM hardware system that we built. As depicted in Fig. \ref{fig:IAF0222}(a), we consider an FRI input signal, referred to as $x(t)$, consisted of two pulses with a width of $100$ns and a delay of $5\mu s$ between them. The sampling kernel mentioned in Section \ref{sec:kernel} was utilized in these experiments. The parameters for the IF-TEM circuit were set to a value of $\kappa=3\cdot 10^{-8}$, with a bias of $b=3V$ and a threshold of $\delta = 1.5V$.
The specific time delays and amplitudes used in this input signal were chosen arbitrarily, and the IF-TEM parameters were selected to adhere to the constraints outlined in \eqref{eq:sample_bound}.
As demonstrated in Fig. \ref{fig:IAF0222}(a), the filtered signal $y(t)$ was transmitted to an IF-TEM sampler, which produced 19 time instances $t_n$, 
resulting in a firing rate of 1.9MHz, which is 4.75 times the rate of innovation and 10.5 times the Nyquist rate. It is important to note that a minimum of $4L+2=10$ time instances are required for off-grid reconstruction.
Fig. \ref{fig:IAF0222}(b) illustrates a comparison between the original input signal and the estimated signal. This demonstrates that the parameters of the FRI system can be robustly estimated while operating at a rate that is 10 times lower than the Nyquist rate.

In Fig. \ref{fig:IAF0223}(a), Fig. \ref{fig:IAF0224}(a), and Fig. \ref{fig:IAF0225}(a), we demonstrate sampling and reconstruction of FRI signals with $L=3,5$ for $h(t)$ as a Dirac impulse and stream of pulses. 
The FRI signal is represented by the green curve, the filtered signal $y(t)$ is shown in yellow, and the time instances $t_n$ produced by the IF-TEM sampler are depicted in blue. In each of these figures, the number of time instances produced is 19, 21, and 22, respectively, resulting in firing rates of 1.9 MHz, 2.1 MHz, and 2.2 MHz, which are all between 9.5 and 10.5 times the Nyquist rate.
The reconstructed FRI signals are shown in Figures \ref{fig:IAF0223}(b), \ref{fig:IAF0224}(b), and \ref{fig:IAF0225}(b). 
The maximum error in time delay estimation is -25 dB. These results indicate that our proposed sampling and reconstruction method is suitable for use in radar and ultrasonic imaging applications. 

Figure \ref{fig:compare} presents a comparison between the reconstruction using the hardware measurements and the simulation for the amplitudes and time delays of the FRI signals with two pulses. This comparison is used to evaluate the performance of the proposed hardware prototype and reconstruction method by comparing the results obtained from the hardware measurements with those obtained from the simulation. The evaluation involves calculating the error between the reconstructed signals obtained from the hardware and simulation, as well as comparing the estimated FRI parameters. This comparison provides insight into the accuracy and reliability of the hardware and reconstruction approach. 
The error in the estimation of the time delay is found to be -25 dB, and this result is consistent with the findings when using $L=3.5$ pulses. 

\begin{figure}
    \centering
    \includegraphics[width=0.5\textwidth]{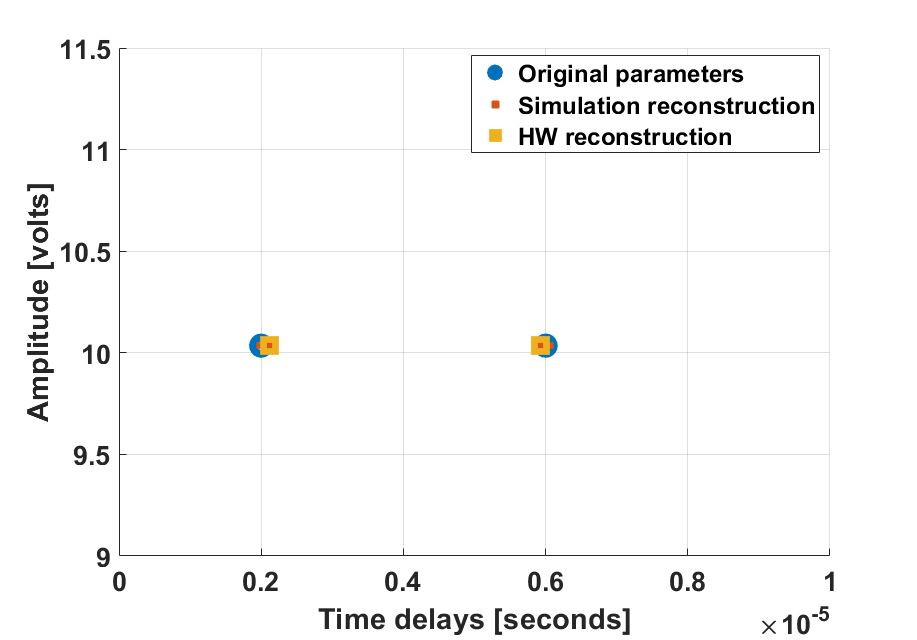}
    \caption{A comparison between the reconstruction using the hardware measurements and the simulation.}
    \label{fig:compare}
\end{figure}

\section{Conclusion}
\label{sec:Conclusions}

In this work, we studied the problem of recovering FRI signals using an IF-TEM sampler. To this end, we introduced a hardware prototype of a sub-Nyquist IF-TEM ADC and developed a robust reconstruction approach to accurately retrieve the FRI parameters. The hardware prototype that we introduced is an asynchronous, energy-efficient ADC that estimates the FRI parameters using a sub-Nyquist framework, which allows it to operate at rates significantly lower than the Nyquist rate.
We have demonstrated that our proposed hardware and reconstruction method can retrieve the FRI parameters with a reconstruction error of up to -25 dB while operating at rates approximately 10 times lower than the Nyquist rate. These results suggest that the proposed hardware prototype and reconstruction approach are effective and efficient in accurately recovering FRI signals and may be useful in various applications, such as radar and ultrasonic imaging. In comparison to traditional ADCs, the proposed prototype is asynchronous and energy-efficient, which may make it particularly attractive for use in energy-constrained systems such as battery-powered devices where these factors are important considerations.

In this study, we investigated the problem of recovering FRI signals using an IF-TEM sampler. To address this challenge, we proposed a hardware prototype of a sub-Nyquist IF-TEM ADC and developed a robust reconstruction approach to accurately retrieve the FRI parameters. The hardware prototype that we introduced is an asynchronous, energy-efficient ADC that estimates the FRI parameters using a sub-Nyquist framework, which allows it to operate at rates significantly lower than the Nyquist rate.
Our proposed hardware and reconstruction method have been demonstrated to be able to retrieve the FRI parameters with a reconstruction error of up to -25 dB while operating at rates approximately 10 times lower than the Nyquist rate. These results suggest that the proposed hardware prototype and reconstruction approach are effective and efficient in accurately recovering FRI signals and may be useful in various applications such as radar and ultrasonic imaging. In comparison to traditional ADCs, the proposed prototype is asynchronous and energy-efficient, which may make it particularly attractive for use in energy-constrained systems such as battery-powered devices where these factors are important considerations.


\bibliographystyle{ieeetr}
\bibliography{refs}

\end{document}